\documentclass[runningheads]{llncs}
\usepackage[T1]{fontenc}   
\usepackage{amsmath}
\usepackage{amssymb}
\usepackage{fancybox}
\usepackage{graphicx}
\usepackage{lmodern}
\usepackage[utf8]{inputenc}

\usepackage{algorithm}

\usepackage{algorithmic}
\newcommand{\junk}[1]{}
\newtheorem{fact}[theorem]{Fact}
\usepackage{subcaption}

\usepackage{microtype}

\begin{document}

\title{$(\min,+)$ Matrix and Vector Products
for Inputs Decomposable into Few Monotone Subsequences}
\titlerunning{$(\min,+)$ Matrix and Vector Products...}
\author{Andrzej Lingas
  \inst{1}
  \and
  Mia Persson
  \inst{2}
}
\authorrunning{A. Lingas and M. Persson}
\institute{
Department of Computer Science, Lund University, 22100 Lund, Sweden. 
\email{Andrzej.Lingas@cs.lth.se}
\and
Department of Computer Science and Media Technology, Malm\"o University, 20506 Malm\"o, Sweden.
\email{Mia.Persson@mau.se}}
\maketitle
\begin{abstract}
We study the time complexity of computing the $(\min,+)$ matrix
  product of two $n\times n$ integer matrices in terms of $n$ and the
  number of monotone subsequences the rows of the first matrix and the
  columns of the second matrix can be decomposed into.  In particular, we
  show that if each row of the first matrix can be decomposed into at
  most $m_1$ monotone subsequences and each column of the second
  matrix can be decomposed into at most $m_2$ monotone subsequences
  such that all the subsequences are non-decreasing or all of them are
  non-increasing then the $(\min,+)$ product of the matrices can
  be computed in $O(m_1m_2n^{2.569})$ time. On the other hand,
  we observe that
  if all the rows of the first matrix are non-decreasing
  and all columns of the second matrix are non-increasing or {\em vice
    versa} then this case is as hard as the general one.
  
  Similarly, we also study the time complexity of computing the
  $(\min,+)$ convolution of two $n$-dimensional integer vectors in
  terms of $n$ and the number of monotone subsequences the two
  vectors can be decomposed into.  We show that if the
  first vector can be decomposed into at most $m_1$ monotone
  subsequences and the second vector can be decomposed into at most $m_2$
  subsequences such that all the subsequences of the first vector are
  non-decreasing and all the subsequences of the second vector are
  non-increasing or {\em vice versa} then their $(\min,+)$
  convolution can be computed in $\tilde{O}(m_1m_2n^{1.5})$ time.  On the
  other, the case when both vectors are non-decreasing
  or both of them are non-increasing is as hard as the general case.
 \end{abstract}

\section{Introduction}
\noindent
{\em $(\min,+)$ matrix product.}
The $(\min,+)$ matrix product problem for two $n\times n$ integer
matrices $A=(a_{i,j})$, $B=(b_{i,j})$ requires computing an $n\times n$
matrix $C=(c_{i,j})$ such that $c_{i,j} = \min\{a_{i,k} + b_{k,j} | 1\le
k\le n\}.$ By the definition, this problem admits an $O(n^3)$-time
  algorithm.  It is known to be equivalent to the fundamental
  all-pairs shortest-paths problem (APSP) \cite{FM71}.  If any of these
  two problems admits an $t(n)$-time algorithm then the other
  problem can be solved in $O(t(n))$ time \cite{FM71}.
  Hence, the APSP hypothesis
  states that solving any of these two problems requires $n^{3-o(1)}$
  time \cite{V18} and the current best algorithm for any of them
  runs in $\frac {n^3} {2^{\Theta (\sqrt {\log n})}}$ time \cite{W04}.

  The $(\min,+)$ matrix product as APSP has a large number of
  important applications.  Because the prospects of deriving
  a substantially
  subcubic upper time bound for the general $(\min,+)$ matrix product
  are so vague, several authors studied the complexity of computing
  this product for restricted integer matrices. Already several
  decades ago, it was known that the $(\min,+)$ matrix product can be
  computed in $O(Mn^{\omega})$ time, when the values of the entries in
  the input matrices are in the range $\{-M,...,M\} \cup \{+\infty \}$
  \cite{AGM97,Yuv76}.  Here,
  $\omega$ stands for
the smallest real number
such that two $n\times n$ matrices can be multiplied using
$O(n^{\omega +\epsilon})$ operations
over the field of reals, for all $\epsilon > 0$
(i.e., the number of
operations is $O(n^{\omega +o(1)})$ \cite{AV21}).
  More recently, one succeeded to derive
  substantially subcubic upper time bounds when, e.g.,: one of the
  matrices has a small number of different entries in each row
  \cite{Y09}, the input matrices are of the so called
  bounded-difference (i.e., all pairs of horizontally and vertically
  adjacent entries differ by at most O(1))\cite{BGS19}, the input
  matrices are geometrically weighted \cite{C10}, one of the matrices
  has a constant approximate rank \cite{VX19}, the entries of one of
  the matrices are of size $O(n)$ and its rows are non-decreasing
  \cite{GPV21}, or just the entries of one of the matrices range over a
  constant number of integers \cite{C10}.

  \noindent
{\em Contributions on $(\min,+)$ matrix product}.
In this paper, we take a more general approach.  We study the
situation when each row of the first matrix $A$ and each column of the
second matrix $B$ can be decomposed into a bounded number of monotone
subsequences. When all the subsequences are non-decreasing or all of
them are non-increasing, we obtain a substantially subcubic algorithm
for the $(\min,+)$ matrix product already when the bound on the number
of monotone subsequences of each row in $A$ and each
column in $B$  is $O(n^{0.215}).$
Namely, our algorithm runs in $O(m_am_bn^{2.569})$
time, where $m_a$ is an upper bound on the number of monotone subsequences
of each row in $A$ and $m_b$ is an upper bound
on the number of the monotone subsequences of each column in $B.$
On the other hand, we
observe that if all the rows of $A$ are non-decreasing
and all columns of $B$  are non-increasing or {\em vice
  versa} then this case is as hard as the general case.
When the entries in each row or column of one of the input matrices
range over $c$ different integers then it is sufficient that
the columns or rows respectively
of the other matrix are decomposable
into at most $n^{0.119}$ just monotone subsequences to subsume
the upper time bound $O(cn^{2.688})$ \cite{C10} (see Fact \ref{fac: 4})
for the case without restrictions on the other
matrix. Our results on $(\min,+)$ matrix
product are summarized in Table 1.

\begin{table*}[h!]
  \label{table: 1}
\begin{center}
\begin{tabular}{|c||c|c|c|c|} \hline
matrix $A$/matrix $B$ & $c_b$ dif. values  & $m_b$
non-decr. subs. & $m_b$
non-incr. subs. 
\\ \hline \hline
$c_a$ different values & $O(c_ac_bn^{\omega})$ & $O(c_am_bn^{2.569})$&
$O(c_am_bn^{2.569})$ 
\\ \hline
$m_a$ non-decr. subs.
& $O(m_ac_bn^{2.569})$ & $O(m_am_bn^{2.569})$  & ? 
\\ \hline
$m_a$ non-incr. subs. 
& $O(m_ac_bn^{2.569})$ & ? & $O(m_am_bn^{2.569})$ 
\\ \hline
arbitrary & $O(c_bn^{2.688})$ \cite{C10}& ? & ? 
\\ \hline
\end{tabular}
\label{table: 1}
\vskip 0.5cm
\caption{Upper time bounds for computing the $(\min ,+)$ matrix
  product of two $n\times n$ integer matrices $A,\ B$,
  where the rows of $A$ and/or
  the columns of $B$ admit decompositions into a bounded
  number of monotone subsequences (in particular non-decreasing
  or non-increasing) or the entries in each row of $A$ or each column of  $B$ range over
  a constant number of integers.}
\end{center}
\end{table*}

\noindent
{\em $(\min,+)$ vector convolution.}
Our approach to the $(\min,+)$ matrix product is in fact similar
to that to $(\min,+)$ convolution of two $n$-dimensional integer
vectors taken by the authors in the prior paper \cite{LP18}.
The $(\min,+)$ convolution problem for two integer vectors
$a=(a_0,...,a_{n-1})$ and $b=(b_0,...,b_{n-1})$
requires computing an $2n-1$ dimensional vector $c=(c_0,...,c_{2n-2})$
such that $c_k= \min\{a_{\ell} + b_{k-\ell}| \ell \in [\max\{k-n+1,0\},
  \min\{ k,n-1\}]\}$ for
$k=0,...,2n-2.$ By the definition, the $(\min,+)$ vector convolution
can be computed in $O(n^2)$ time but again
getting any substantially subquadratic upper time bound for
this problem would be a breakthrough. The $(\min,+)$ vector convolution
has also a large number of important applications
ranging from
stringology to knapsack problem \cite{BC22,BC14,Cyg,M96}.

\noindent
{\em Contributions on $(\min,+)$ convolution.}
We correct the requirements on the monotonicity of the vector subsequences in
the statement of Theorem 3.7 in \cite{LP18} and provide a proof of the
corrected theorem. It states that the $(\min,+)$
  convolution of two $n$-dimensional integer vectors $a$ and $b$,
  given with the decompositions of the sequences of their consecutive
  coordinates into $m_a$ and $m_b$ subsequences respectively such that
  either all the subsequences of $a$ are non-decreasing and all the
  subsequences of $b$ are non-increasing or {\em vice versa}, can be
  computed in $\tilde{O}(m_am_bn^{1.5})$ time.
   On the
  other hand, the case when both vectors are non-decreasing
  or both of them are non-increasing is as hard as the general case.
  Table 2
  summarizes the updated results on $(\min,+)$ vector convolution
  (cf. \cite{LP18}).
\begin{table*}[h!]
\label{table: 1}
\begin{center}
\begin{tabular}{|c||c|c|c|c|} \hline
vector $a$/vector $b$ & $c_b$ dif. values  & $m_b$
non-decr. subs. & $m_b$
non-incr. subs. 
\\ \hline \hline
$c_a$ different values & $\tilde{O}(c_ac_bn)$ & $\tilde{O}(c_an^{1.5})$&
$\tilde{O}(c_an^{1.5})$ 
\\ \hline
$m_a$ non-decr. subs.
& $\tilde{O}(c_bn^{1.5})$ & ? & $\tilde{O}(m_am_bn^{1.5})$   
\\ \hline
$m_a$ non-incr. subs. 
& $\tilde{O}(c_bn^{1.5})$ & $\tilde{O}(m_am_bn^{1.5})$ & ?
\\ \hline
arbitrary & $\tilde{O}(c_bn^{1.5})$ & ? & ? 
\\ \hline
\end{tabular}
\vskip 0.5cm
\caption{Upper time bounds for computing the $(\min,+)$ convolution of
two $n$-dimensional integer vectors either with coordinates  
having a bounded number of different values, or decompositions into a 
number of non-decreasing or non-increasing subsequences.}
\end{center}
\end{table*}

\noindent
  {\em Techniques.}
  Our algorithms for the $(\min,+)$ matrix product as well as those
  for the $(\min,+)$ vector convolution are mostly based on efficient
  reductions to collections of maximum or/and minimum witness
  problems for corresponding Boolean matrix products or Boolean vector
  convolutions, respectively.  One of our algorithms uses directly a
  method similar to that known for the extreme witnesses.  For the
  definition of the extreme witness problems and facts on them see
  Preliminaries.

  \noindent
      {\em Paper organization.} The next section contains basic
      definitions and facts. Section 3 presents our results on
      $(\min,+)$ matrix product while Section 4 presents our results
      on $(\min,+)$ vector convolution.

\section{Preliminaries}
For two $n$-dimensional vectors 
$a=(a_0,...,a_{n-1})$ and $b=(b_0,...,b_{n-1})$
over a semi-ring $(\mathbb{U},\oplus, \odot)$, 
their convolution over the semi-ring is a vector
\newline
$c=(c_0,...,c_{2n-2})$, where
$c_i=\bigoplus_{l=\max\{i-n+1,0\}}^{\min\{ i,n-1\}}a_l\odot b_{i-l}$
for $i=0,...,2n-2.$ 
Similarly, for a $p\times q$ matrix $A$
and a $q\times r$ matrix $B$ over the semi-ring,
their matrix product over the semi-ring is a $p\times r$ matrix
$C=(c_{i,j})$ such that $c_{i,j}=\bigoplus_{m=1}^{q}a_{i,m}\odot b_{m,j}$
for $1\le i\le p$ and $1\le j\le r.$
In particular, for the semi-rings $(\mathbb{Z},+,\times),$
$(\mathbb{Z},\min,+),$ $(\mathbb{Z},\max,+),$ and $(\{0,1\},\vee, \wedge)$, we obtain
the arithmetic, $(\min,+),$ $(\max,+),$ and the Boolean convolutions
or matrix products,
respectively.

We shall use the unit-cost RAM computational model
with
computer word of length logarithmic in the maximum of
the size of the input and the value of the largest input integer.

For a positive integer $r,$ we shall denote
the set of positive integers not greater than $r$
by $[r].$

Consider the Boolean matrix product $C=(c_{i,j})$ of Boolean $n\times
n$ matrices $A=(a_{i,j})$ and $B=(b_{i,j})$.  A \textit{witness} for a
non-zero entry $c_{i,j}$ of the product $C$ is any index $k\in [n]$
such that $a_{i,k}$ and $b_{k,j}$ are equal to $1.$ Such a minimum
index is the \textit{minimum witness} for $c_{i,j}$ while such a
maximum index is the the \textit{maximum witness} for $c_{i,j}$.  The
\textit{minimum witness problem} (\textit{maximum witness problem},
respectively) is to report the minimum  witness (maximum witness,
respectively) for each non-zero entry of the Boolean matrix product of
the two input matrices.

For positive real numbers
$p,\ q,\ s,$ $\omega(p,q,s)$ denotes the smallest
real number
such that an $n^p\times n^q$ matrix can
be multiplied by $n^q\times n^s$ matrix using
$O(n^{\omega(p,q,s) +\epsilon})$ operations
over the field of reals, for all $\epsilon > 0.$
For convenience, $\omega$ stands for  $\omega(1,1,1).$

  \begin{fact}\cite{CK07}\label{fac: minwit}
    The minimum witness problem and the maximum witness problem
    for the Boolean matrix product of two Boolean $n\times n$ matrices
    can be solved in $O(n^{2+\lambda})$ time, where $\lambda $
    satisfies the equation $\omega(1, \lambda, 1) = 1 + 2\, \lambda $.
    By currently best  bounds on $\omega(1,\lambda,1)$, $O(n^{2+\lambda})=O(n^{2.569}).$
 \end{fact}

The currently best bounds on $\omega(1,\lambda,1)$ follow from a fact
in \cite{HP} combined with the recent improved estimations on the
parameters $\omega=\omega(1,1,1)$ and $\alpha$, see
\cite{LG14,LGU}. They yield an $O(n^{2.569})$ upper bound on the
running time of the algorithm for minimum and maximum witnesses in
\cite{CK07} (originally, $O(n^{2.575})$).

The following fact is well known (cf. \cite{FP74}).

\begin{fact}\label{fac: 1}
Let  $p$ and $q$ be two $n$-dimensional 
integer vectors.
The arithmetic convolution of $p$ and $q$ can be computed in
$\tilde {O}(n)$ time. Hence, also the Boolean convolution
of two $n$-dimensional vectors can be computed in
$\tilde {O}(n)$ time.
\end{fact}

Let $c=(c_0,...,c_{2n-2})$ be the Boolean convolution of two
$n$-dimensional Boolean vectors $a$ and $b.$ A {\em witness} of
$c_i=1$ is any $l\in [ \max\{i-n+1,0\}, \min\{ i,n-1\}]$
such that $a_l\wedge b_{i-l}=1.$ 
A {\em minimum witness} (or {\em maximum witness}) of
$c_i=1$ is the smallest (or, the largest, respectively) witness
of $c_i.$
The 
{\em minimum witness problem}, or {\em maximum witness problem}
for the  Boolean convolution 
of two 
$n$-dimensional Boolean vectors is to determine  
the minimum witnesses or
the maximum witnesses, respectively, for all non-zero entries
of the Boolean convolution of the vectors.

\begin{fact}(Theorem 3.2 in \cite{LP18})\label{fac: minwitcon}
The minimum witness problem
(maximum witness problem, respectively) for the Boolean convolution
of two $n$-dimensional vectors can be solved in
$\tilde{O}(n^{1.5})$ time.
\end{fact}
For a sequence $s$ of integers, we shall denote 
the minimum number of monotone
subsequences into which $s$ can be decomposed by $mon(s)$.

\begin{fact}\label{fac: 2} \cite{FKN02,YCL07}. A sequence $s$ of $n$ integers
can be decomposed into 
\newline
$O(mon(s)\log n)$ monotone subsequences in $O(n^{1.5}\log n)$ time. 
\end{fact}

\begin{fact} \label{fac: 4}(Theorem 3.2 in \cite{C10}).
Let  $A$ and $B$ be two $n\times n$
integer matrices,  
where the entries of one of the matrices
range over at most $c$ different integers.
The $(\min,+)$ matrix product of $A$
and $B$ can be computed in
$O(cn^{2.688})$ time.
\end{fact}

\section{(min,+) matrix product}

Consider two $n\times n$ integer matrices $A$ and $B.$
If we are given decompositions of the rows
of $A$ and the columns of $B$
into monotone subsequences
such that either all the subsequences 
are non-decreasing or all of them are non-increasing then
we can use the algorithm depicted in Fig. \ref{fig: algo1}
in order to compute the $(\min,+)$ matrix product of $A$ and $B.$
First, for all $i,j \in [n],$ for
each subsequence $a^o_i$ of
the $i$-th row of $A$ and
each subsequence $b^r_j$ of
the $j$-th column of $B,$ we compute
the Boolean vectors $char(a^o_i)$ and $char(b^r_j)$
indicating with ones 
the entries  of the row and column covered by $a^o_i$ or
$b^r_j,$ respectively. Next, we form the Boolean matrices $A^o$
whose rows are the vectors $char(a^o_i)$ and the  Boolean matrices
$B^r$ whose columns are the vectors $char(b^r_j)$.
Then , depending
if the subsequences are non-decreasing
or non-increasing, for each pair of matrices $A^o,\ B^r,$ we compute either
 the minimum witnesses of the Boolean matrix product of $A^o$
 and $B^r$  or the maximum witnesses of this Boolean product,
 respectively.
 We use the extreme witnesses to update
the current entries of the computed $(\min,+)$
matrix product of $A$ and $B.$
The correctness of the reduction
to extreme witnesses for the Boolean matrix product
of $A^o$ and $B^r$ in the algorithm is implied by
the following observation.

\begin{figure}
  \begin{algorithmic}[1]
    {\small
\REQUIRE two $n\times n$ integer matrices
$A=(a_{i,j})$ and 
$B=(b_{i,j})$, for each $i\in [n]$,  a
decomposition of the $i$-th row of $A$ into 
$m_a$ subsequences $a^o_i$,
and for each $j\in [n],$
a decomposition of the $j$-th column of $B$ 
 into
 $m_b$ subsequences $b^r_j$,
 such that either  all the subsequences are non-decreasing
 or all of them are non-increasing.
\ENSURE the $(\min,+)$ matrix product  $C= (c_{i,j})$ of 
$A$ and $B$.
\FOR {each $o\in [m_a]$}
\FOR {each $i\in [n]$}
\STATE form a Boolean vector $char(a^o_i)$
with $n$ coordinates indicating with
ones the entries of
the $i$-th row of $A$ covered by $a^o_i$
\ENDFOR
\STATE form a Boolean matrix $A^o$, where
for $i\in [n],$ $char(a^o_i)$ is the $i$-th row
\ENDFOR
\FOR {each $r\in [m_b]$}
\FOR {each $j\in [n]$}
\STATE form a Boolean vector $char(b^r_j)$
with $n$ coordinates indicating with
ones the entries  of the $j$-th column of $B$
covered by $b^r_j$
\ENDFOR
\STATE form a Boolean matrix $B^r$, where
for $j\in [n],$ $char(b^r_j)$ is the $j$-th column
\ENDFOR
\STATE initialize the $C=(c_{i,j})$ matrix by setting each $c_{i,j}$
to $+\infty$
\FOR {each pair $A^o$ and $B^r$}
\STATE {\bf if} the subsequences are non-decreasing {\bf then} compute the 
minimum witnesses $(wit(d_{i,j}))$
for the Boolean matrix product $(d_{i,j})$  of $A^o$ and $B^r$
\STATE {\bf if} the subsequences are non-increasing {\bf then} compute the 
maximum witnesses $(wit(d_{i,j}))$
for the Boolean matrix product $(d_{i,j})$  of $A^o$ and $B^r$
 \FOR {$i = 1$ \TO $n$}
\FOR {$j = 1$ \TO $n$}
\STATE {\bf if} $d_{i,j}\neq 0$ {\bf then} $c_{i,j}\gets \min \{
a_{i,wit(d_{i,j})}+b_{wit(d_{i,j}),j},c_{i,j}\}$
\ENDFOR
\ENDFOR
\ENDFOR
\STATE $C \gets (c_{i,j})$
\RETURN $C$}
\end{algorithmic}
\caption{An algorithm for computing the $(\min,+)$ product of
  two $n\times n$ integer matrices $A$ and $B$
  given with decompositions of all  rows of $A$ into
  $m_a$ subsequences and decompositions of all columns
  of $B$ into $m_b$ subsequences such
  that either all the subsequences are non-decreasing or all
  the subsequences are non-increasing.
}
\label{fig: algo1}
\end{figure}
\vfill
\newpage
\begin{remark}\label{rem: second}
  Let  $A=(a_{i,j})$ and $B=(b_{i,j})$ be two $n\times n$ integer matrices.
  Next, let $a'$ be a subsequence of the sequence of entries in an
  $i$-th row of $A$ and let $b'$ a subsequence of the sequence of
  entries in an $j$-th column of $B.$
  If $a'$ and $b'$  are non-decreasing then
  if the set $\{ a_{i,k}+b_{k,j}| k\in [n] \wedge a_{i,k}\in a'
  \wedge b_{k,j}\in b'\}$ is not empty
  then the minimum sum in the set
  is achieved by the pair minimizing the index $k$.
  Analogously, if $a'$ and $b'$ are non-increasing
  then the minimum sum is achieved by the pair
  maximizing the index $k$.
  \end{remark}

\begin{theorem}\label{theo: mat}
  Let  $A=(a_{i,j})$ and $B=(b_{i,j})$ be two $n\times n$ integer matrices.
  Suppose that for each $i,j\in [n],$ there is given
  a decomposition of the $i$-th row of $A$
  into at most $m_a$ subsequences  and a decomposition of the $j$-th
  column of $B$ into at most $m_b$ subsequences, where either all the subsequences
  are non-decreasing or all of them are non-increasing.
  Then, the $(\min, +)$  product of $A$ and $B$ can be computed
  in $O(m_am_b n^{2.569})$ time.
\end{theorem}
\begin{proof}
  By Remark \ref{rem: second}, the following condition holds:
  (*) if $d_{i,j}\neq 0$ in line 19 of the
  algorithm in Fig. \ref{fig: algo1} then $\min
  \{a_{i,k}+b_{k,j}|k\in [n] \wedge char(a^o_i)_k=1 _i \wedge char(b^r_j)_k=1  \}$ is equal to
  the first argument of the minimum in this line, i.e.,
  $a_{i,wit(d_{i,j})}+b_{wit(d_{i,j}),j}$. Hence,
  none of the entries of the output matrix $C$
has a lower value than the corresponding entry  of
the $(\min,+)$ matrix product of $A$ and $B.$
Conversely, if the ${i,j}$ entry of
the $(\min,+)$ matrix  of $A$ and $B$ equals
$a_{i,k}+b_{k,j}$ then there exist
$o,\ r$ such that  $a_{i,k}\in a^o_i $ and $b_{k,j}\in b^r_j ,$
i.e., more precisely $char(a^o_i)_k=1$ and $char(b^r_j)_k=1.$
Hence, again by (*) and line 19 in the algorithm,
the $c_{i,j}$ entry in the output matrix  has
value not larger than the corresponding $i,j$ entry
of the $(\min,+)$ matrix product of $A$ and $B.$

The time complexity of the algorithm is dominated by
the $m_am_b$ computations of minimum or maximum witnesses
of the Boolean product of two Boolean $n\times n$ matrices.
Thus, by Fact \ref{fac: minwit} the algorithm
runs in  $O(m_am_bn^{2.569})$ time.
\qed
\end{proof}
\par
\noindent
    {\em  Example 1.}
    We shall assume the notation from our first algorithm.
    Suppose that two input integer matrices $A$ and $B$
    have size $6\times 6$ and that each row of $A$
    can be decomposed into at most $3$
    non-decreasing subsequences while each column of
    of $B$ can be decomposed into at most $2$
    non-decreasing subsequences. Suppose in particular
    that the fourth row $a_4$ of $A$ is $(1,7,3,9,8,4)$ while the fifth
    column $b_5$ of $B$ is $(5,11,2,7,13,10).$ Then, it is easy to see
    that in the $(\min,+)$ product $(c_{i,j})$ of $A$ and $B$,
    $c_{4,5}= 5$ holds. Note that $a_4$ can be decomposed
    into three following non-decreasing subsequences
    $a^1_4=(1,\ ,3,\ ,\ ,4) ,$ $a^2_4=(\ ,7,\ ,\ ,8,\ ),$
    $a^3_4=(\ ,\ ,\ ,9,\ ,\ )$ while $b_5$
    can be decomposed into two non-decreasing subsequences
    $b^1_5=(5,11,\ ,\ ,13,\ )$ and $b^2_5=(\ ,\ ,2,7,\ ,10)$. Their characteristic
    Boolean vectors are $char(a^1_4)=(1,0,1,0,0,1),$
    $char(a^2_4)=(0,1,0,0,1,0),$ $char(a^3_4)=(0,0,0,1,0,0)$,
    and $char(b^1_5)=(1,1,0,0,1,0),$ $char(b^2_5)=(0,0,1,1,0,1),$ respectively.
    For $o\in [3]$ and $r\in [2],$ the inner Boolean
    product of the vectors $char(a^o_4)$ and
    $char(b^r_5)$ yields the $d_{4,5}$ entry
    of the Boolean matrix product of the Boolean matrices
    $A^o$ and $B^r$ in the algorithm.
    The minimum witness of the 
    entry  $d_{4,5}$  is $1$ for $o=1,\ r=1,$
    $3$ for $o=1,\ r=2,$ $2$ for $o=2,\ r=1,$ and
    $4$ for $o=3,\ r=2,$ respectively.
    For the other combinations of $o$ and $r,$
    it is undefined.
    Hence, $c_{4,5}$ is computed as
    the minimum of  $1+5,3+2, 7+11,9+7$  which is $5$ as required. 
    \par
    \vskip 3pt
We shall call a sequence of integers {\em uniform}
if all its elements have the same value.

A uniform subsequence of a matrix row or column covering all entries in the row
or column having the same fixed value is both non-increasing and
non-decreasing. If the entries in the row or column can have
at most $c$ different values then the row or column can be easily decomposed
into at most $c$ uniform subsequences.
Hence, if the entries in rows or columns of one of the input matrices
range over relatively few different integers then
it is sufficient to decompose the rows or columns of
the other matrix into relatively few monotone subsequences
in order to obtain relatively efficient algorithm for
the $(\min,+)$ matrix product. The aforementioned subsequences
do not have to be simultaneously non-decreasing or
non-increasing as the counterpart subsequences in
the first matrix are uniform and hence are both
non-decreasing and non-increasing.
\begin{figure}[bt!]
\begin{algorithmic}[1]
\REQUIRE two $n\times n$ integer matrices
$A=(a_{i,j})$ and 
$B=(b_{i,j})$, for each $i\in [n]$  a
decomposition of the $i$-th row of $A$ into 
$m$ monotone subsequences $a^o_i$ 
and for each $j\in [n],$
a decomposition of the $j$-th column of $B$ 
 into 
 $c$ uniform subsequences $b^r_j$.
\ENSURE the $(\min,+)$ matrix product  $C= (c_{i,j})$ of 
$A$ and $B$ 
\FOR {each $o\in [m_a]$}
\FOR {each $i\in [n]$}
\STATE form a Boolean vector $char(a^o_i)$
with $n$ coordinates indicating with
ones the entries of
the $i$-th row of $A$ covered by $a^o_i$
\ENDFOR
\STATE form a Boolean matrix $A^o$, where
for $i\in [n],$ $char(a^o_i)$ is the $i$-th row
\ENDFOR
\FOR {each $r\in [m_b]$}
\FOR {each $j\in [n]$}
\STATE form a Boolean vector $char(b^r_j)$
with $n$ coordinates indicating with
ones the entries  of the $j$-th column of $B$
covered by $b^r_j$
\ENDFOR
\STATE form a Boolean matrix $B^r$, where
for $j\in [n],$ $char(b^r_j)$ is the $j$-th column
\ENDFOR
\STATE initialize the matrix $C=(c_{i,j})$ by setting
each $c_{i,j}$ to $+\infty$
\FOR {each pair $A^o,\ B^r$}
\STATE compute the 
minimum witnesses $(minwit(d_{i,j}))$  and the maximum witnesses
$(maxwit(d_{i,j}))$
of the Boolean matrix product $(d_{i,j})$  of $A^o$ and $B^r$
 \FOR {$i = 1$ \TO $n$}
\FOR {$j = 1$ \TO $n$}
\STATE {\bf if} $d_{i,j}\neq 0$
and $a^o_i$ is non-decreasing {\bf then} $c_{i,j}\gets \min \{
a_{i,minwit(d_{i,j})}+b_{minwit(d_{i,j}),j},c_{i,j}\}$
\STATE {\bf if} $d_{i,j}\neq 0$
and $a^o_i$ is non-increasing {\bf then} $c_{i,j}\gets \min \{
a_{i,maxwit(d_{i,j})}+b_{maxwit(d_{i,j}),j},c_{i,j}\}$
\ENDFOR
\ENDFOR
\ENDFOR
\STATE $C \gets (c_{i,j})$
\RETURN $C$
\end{algorithmic}
\caption{An algorithm for computing the $(\min,+)$ product of
  two $n\times n$ integer matrices $A$ and $B$
  given with decompositions of each row of $A$ into
  $m$ monotone subsequences and decompositions of each column
  of $B$ into $c$ uniform subsequences.
}
\label{fig: algo2}
\end{figure}

\begin{remark}\label{rem: third}
  Let  $A=(a_{i,j})$ and $B=(b_{i,j})$ be two $n\times n$ integer matrices.
  Next, let $a'$ be a subsequence of the sequence of entries in an
  $i$-th row of $A$ and let $b'$ a uniform subsequence of
  the sequence of entries in an $j$-th column of $B.$
  If $a'$ is non-decreasing and the set $\{ a_{i,k}+b_{k,j}| k\in [n] \wedge a_{i,k}\in a'
  \wedge b_{k,j}\in b'\}$ is not empty
  then the minimum sum in the set
  is achieved by the pair minimizing the index $k$.
  Analogously, if $a'$ is non-increasing
  then the minimum sum is achieved by the pair
  maximizing the index $k$.
\end{remark}

Remark \ref{rem: third} and the algorithm
depicted in Fig. \ref{fig: algo2} yield the following theorem.

\begin{theorem}\label{theo: mat1}
  Let  $A=(a_{i,j})$ and $B=(b_{i,j})$ be two $n\times n$ integer matrices.
  Suppose that at least one of the two following conditions holds:
  \begin{enumerate}
  \item the entries in each column of $B$ range over at most $c$ integers
and for each $i \in [n],$ there is given
  a decomposition of the $i$-th row of $A$
  into at most $m$ monotone subsequences;
\item
 the entries in each row of $A$ range over at most $c$ integers
and for each $j\in [n],$ there is given
  a decomposition of the $j$-th column of $B$
  into at most $m$ monotone subsequences,
  \end{enumerate}
  Then, the $(\min, +)$  product of $A$ and $B$ can be computed
  in $O(mc n^{2.569})$ time.
\end{theorem}
\begin{proof}
  For a square matrix $D,$ let $D^t$ denote its transpose.
  $(AB)^t=B^tA^t$ holds. Therefore, we may assume w.l.o.g. that
  the first condition in the theorem statement holds.
  Since then the entries in each column of $B$ range over at most $c$
  different values we can easily decompose each column
  of $B$ into $c$ uniform subsequences (some can be empty)
  in $O(cn^2)$ time in total. We can also fill the decompositions
  of the rows of $A$ to exactly $m$ monotone subsequences
  by adding empty subsequences.
  Hence, we can apply Algorithm 2.
  It remains to show the correctness of the algorithm
  and estimate its time complexity.

  By Remark \ref{rem: third}, the following condition holds: (**) if
  $a^o_i$ is non-decreasing
  and $d_{i,j}\neq 0$ in line 18 of the
  algorithm in Fig. \ref{fig: algo2} then $\min
  \{a_{i,k}+b_{k,j}|k\in [n] \wedge char(a^o_i)_k=1 \wedge char(b^r_j)_k=1 \}$ is equal to
  the first argument of the minimum in this line, i.e.,
  $a_{i,minwit(d_{i,j})}+b_{minwit(d_{k,j}),j}.$
Also, if $a^o_i$ is non-increasing and
and $d_{i,j}\neq 0$ in line 19 of the
algorithm in Fig. \ref{fig: algo2} then
$\min\{a_{i,k}+b_{k,j}|k\in [n] \wedge char(a^o_i)_k=1 \wedge char(b^r_j)_k=1 \}$
   is equal to
  the first argument of the minimum in this line, i.e.,
  $a_{i,maxwit(d_{i,j})}+b_{maxwit(d_{i,j}),j}.$
  Hence,
  none of the entries of the output matrix $C$
has a lower value than the corresponding entry  of
the $(\min,+)$ matrix product of $A$ and $B.$
Conversely, if the ${i,j}$ entry of
the $(\min,+)$ matrix  of $A$ and $B$ equals
$a_{i,k}+b_{k,j}$, where $k\in[n],$ then there exist
$o,\ r$ such that  $char(a^o_i)_k=1$ and $char(b^r_j)_k=1.$
Hence, again by (**) and lines  18, 19 in the algorithm,
the $c_{i,j}$ entry in the output matrix  has
value not larger than the corresponding entry
of the $(\min,+)$ matrix product of $A$ and $B.$

The time complexity of the algorithm is dominated by
the $mc$ computations of minimum and maximum witnesses
of the Boolean product of two Boolean $n\times n$ matrices.
Thus, by Fact \ref{fac: minwit} the algorithm runs in  $O(mcn^{2.569})$ time.
\qed \end{proof}

Recall that for an integer sequence $s,$ $mon(s)$ denotes
the minimum number of monotone subsequences into which $s$
can be decomposed.
By Fact \ref{fac: 2}, we obtain immediately the following corollary from
Theorem \ref{theo: mat1}.

\begin{corollary}\label{cor: mm}
  Let $A$ and $B$ be two $n\times n$ integer matrices.
  Let $m_1$ be the maximum of $mon(d)$ over
  all sequences $d$ formed by consecutive
  entries in the rows of $A$ 
  and let $m_2$ be the maximum of $mon(d)$ over
  all sequences $d$ formed by consecutive
  entries in the columns of $B.$
  If the entries in each row of $A$ range over at most $c_a$ different
  integers then the $(\min,+)$ matrix product
  of $A$ and $B$ can be computed in $O(m_2c_an^{2.569}\log n)$
  time. Similarly, if the entries in each column of $B$ range over at
  most $c_b$ different integers than the product
  can be computed in $O(m_1c_bn^{2.569}\log n)$ time.
\end{corollary}

Note that when $m_1$ or $m_2$ does not exceed $n^{0.119}$ then the
upper bound of Corollary \ref{cor: mm}
subsumes that of Fact \ref{fac: 4}.

We have also the following observation.

\begin{remark}
Let  $A$ and $B$ be two $n\times n$ integer matrices
such the values of the entries in each row of $A$ range over
at most $c_a$
  integers while the entries in each column of $B$ range over
  at most $c_b$ different integers.
  The $(\min, +)$  product of $A$ and $B$ can be computed
  by an immediate reduction to $c_ac_b$ Boolean matrix products
  of $n\times n$ matrices and thus it can be computed in $O(c_ac_bn^{\omega})$ time.
\end{remark}

Finally, we demonstrate
that the case when the rows of the first matrix are non-decreasing
and the columns of the second
matrix are non-increasing or {\em vice versa}
is as hard as the general case.

\begin{theorem}
  The problem of computing the $(\min,+)$ matrix product
  of two $n\times n$ integer matrices $A=(a_{i,j})$ and $B=(b_{i,j})$,
  where for $i\in [n],$ the rows $a_{i,1},...,a_{i,n}$  of $A$ are non-decreasing
  and the columns $b_{1,j},...,b_{n,j}$ of $B$ are non-increasing
  or {\em vice versa} is equally hard as computing the product
  for arbitrary $n\times n$ integer matrices.
\end{theorem}
\begin{proof}
  Let $M$ be the maximum absolute value of an entry in the matrices $A,\ B.$
  Transform the matrix $A$ to a matrix $A'$ by setting
  $a'_{i,k}=a_{i,k} +2kM$ for $i,\ k \in [n].$ Observe that
  each row in $A'$ is non-decreasing. Similarly, define
  the matrix $B'$ by setting $b'_{k,j}=b_{k,j} -2kM$
  for $j,\ k \in [n].$ Similarly observe that each column
  of $B'$ is non-increasing.
  Now, consider the $(\min,+)$ matrix products
  $C=(c_{i,j})$ of $A,\ B$ and $C'=(c'_{i,j})$ of $A',\ B'$.
  For $i,j \in [n],$ we have
  $$c'_{i,j}=\min\{ (a_{i,k}+2kM)+(b_{k,j}-2kM)|k\in [n]\}\\=c_{i,j}.$$
  The proof for the case where the rows of
  the first matrix are non-increasing
  and the columns of
  the second matrix are non-decreasing  is symmetric.
  \qed
\end{proof}
We summarize our results on the $(\min,+)$ matrix product
in Table 1.

\section{$(\min,+)$ convolution}

If we are given decompositions of the two input $n$-dimensional
vectors $a$ and $b$ into monotone subsequences
such that either all the subsequences of $a$
are non-decreasing and all the subsequences of $b$ are non-increasing
or {\em vice versa} then
we can use the algorithm depicted in Fig. \ref{fig: algo3}
in order to compute the $(\min,+)$ convolution of $a$ and $b.$
First, for each subsequence $a^i$ of $a$ and
each subsequence $b^j$ of $b,$ we compute
the Boolean vectors $char(a^i)$ and $char(b^j)$
indicating with ones 
the coordinates of $a$ or $b$ covered by $a^i$ or
$b^j,$ respectively. Next, depending
if the subsequences are non-decreasing
and non-increasing, respectively, or
{\em vice versa}, for each pair of
such subsequences $a^i$ and $b^j$, we compute
the minimum witnesses of the Boolean convolution
of $char(a^i)$ and $char(b^j)$ or the maximum witnesses of this Boolean convolution,
respectively. We use the extreme witnesses to update
the current coordinates of the computed $(\min,+)$
convolution.
The correctness of  the reduction
to extreme witnesses of the Boolean convolution
of $char(a^i)$ and $char(b^j)$ in the algorithm is implied by
the following observation.

\begin{remark}\label{rem: first}
  Let $a=(a_0,...,a_{n-1})$ and $b=(b_0,...,b_{n-1})$
  be two $n$-dimensional integer vectors.
  Next, let $a'$ be a subsequence of $a_0,...,a_{n-1}$ and
  let $b'$ be a subsequence of $b_0,...,b_{n-1}.$ For each $k \in \{0,...,2n-2\}$,
  if $a'$ is non-decreasing and $b'$ non-increasing then
  if the set $\{ a_{\ell}+b_{k-\ell}| a_{\ell}\in a'
  \wedge b_{k-\ell}\in b'\}$ is not empty
  then the minimum sum in the set
  is achieved by a pair minimizing the index $\ell$ (thus maximizing
  $k-\ell $). Analogously, if $a'$ is non-increasing and $b'$
  is non-decreasing then the minimum sum is achieved by a pair
  maximizing the index $\ell$ (thus, minimizing the index $k-\ell$).
  \end{remark}

Hence, we obtain the following theorem, correcting Theorem 3.7 in \cite{LP18}.

\begin{figure}[bt!]
  \begin{algorithmic}[1]
{\small
\REQUIRE two $n$-dimensional
vectors $a=(a_0,...,a_{n-1})$ and 
$b=(b_0,...,b_{n-1})$
with integer coordinates and their
decompositions into $m_a$ and $m_b$ 
subsequences $a^i$
and $b^j$
respectively such that either all the subsequences $a^i$
are non-decreasing and all the subsequences $b^j$
are non-increasing or {\em vice versa}.
\ENSURE the $(\min,+)$ convolution $c = (c_0,....,c_{2n-2})$ of 
$a$ and $b$. 
\FOR {each $a^i$}
\STATE form a Boolean vector $char(a^i)$
with $n$ coordinates indicating with
ones the coordinates of $a$ covered by $a^i$
\ENDFOR
\FOR {each $b^j$}
\STATE form a Boolean vector $char(b^j)$
with $n$ coordinates indicating with
ones the coordinates of $b$ covered by $b^j$
\ENDFOR
\STATE initialize the vector $c=(c_0,...,c_{2n-2})$ by setting
all its coordinates to $+\infty $
\FOR {each pair $a^i,\ b^j$}
\STATE {\bf if} the subsequence $a^i$ is non-decreasing {\bf then} compute the 
minimum witnesses $wit(d_0),$...,$wit(d_{2n-2})$
of the Boolean convolution $(d_0,...,d_{2n-2})$ of $char(a^i)$ and $char(b^j)$
\STATE {\bf if} if the subsequence $a^i$ is non-increasing {\bf then} compute the 
maximum witnesses $wit(d_0),$...,$wit(d_{2n-2})$
of the Boolean convolution $(d_0,...,d_{2n-2})$ of $char(a^i)$ and $char(b^j)$
\FOR {$k = 0$ \TO $2n-2$}
\STATE {\bf if} $d_k\neq 0$ {\bf then} $c_k\gets \min \{
a_{wit(d_k)}+b_{k-wit(d_k)},c_k\}$
\ENDFOR
\ENDFOR
\STATE $c \gets (c_0,...,c_{2n-2})$
\RETURN $c$}
\end{algorithmic}
\caption{An algorithm for computing the $(\min,+)$ convolution
$c$ of two $n$-dimensional integer vectors $a$ and $b$
given with their
decompositions into $m_a$ and $m_b$ 
subsequences respectively such that either all the subsequences of $a$
are non-decreasing and all the subsequences of $b$ are non-increasing
or {\em vice versa}.
}
\label{fig: algo3}
\end{figure}

\begin{theorem}
Let $a$ and $b$ be two $n$-dimensional integer vectors
given with the decompositions of
the sequences 
of their consecutive coordinates into $m_a$ and $m_b$ 
monotone subsequences respectively such that
either all the subsequences of $a$
are non-decreasing and all the subsequences of $b$ are non-increasing
or {\em vice versa}.
The algorithm depicted in Fig. \ref{fig: algo3}
computes the $(\min,+)$ convolution of $a$ and $b$
in $\tilde{O}(m_am_bn^{1.5})$
time.
\end{theorem}
\begin{proof}
The proof of the correctness of the algorithm depicted in
  Fig. \ref{fig: algo3} is analogous to that of the correctness of the
  algorithm depicted in Fig. \ref{fig: algo1}.  In
  particular, we obtain the following implication 
  from Remark \ref{rem: first}: (***) if $d_k\neq 0$ in line 12 of the
  algorithm in Fig. \ref{fig: algo3} then $\min
  \{a_{\ell}+b_{k-\ell}| char(a^i)_{\ell}=1  \wedge char(b^j)_{k-\ell}=1 \}$ is equal to
  the first argument of the minimum in this line, i.e.,
  $a_{wit(d_k)}+b_{k-wit(d_k)}.$
  Hence,
  none of the coordinates of the output vector
has a lower value than the corresponding coordinate of
the $(\min,+)$ convolution of $a$ and $b.$
Conversely, if the $k$ coordinate of
the $(\min,+)$ convolution of $a$ and $b$ equals
$a_{\ell}+b_{k-\ell}$ then there exists
$i,j$ such that  $char(a^i)_{\ell}=1$ and $char(b^j)_{k-\ell}=1$.
Hence, again by (***) and line 12 in the algorithm,
the $c_k$ coordinate in the output vector has
value not larger than the corresponding coordinate
of the $(\min,+)$ convolution of $a$ and $b$. 

The time complexity analysis of the algorithm in Fig. \ref{fig: algo3}
is also similar to that of the algorithm in Fig. \ref{fig: algo1}.
It is dominated by the $m_am_b$ runs of the $\tilde{O}(n^{1.5})$-time
algorithm for 
the extreme witnesses of the
Boolean convolution of two $n$-dimensional Boolean vectors
given in Fact \ref{fac:  minwitcon}.
\qed
\end{proof}
\par
\noindent
{\em  Example 2.}
We shall assume the notation from the first algorithm
in this section. Suppose that
$a=(a_0,...,a_5)=(1,7,3,9,8,4)$ and $b=(b_0,...,b_5)=(13,7,11,5,10,12).$
 Then, it is easy to see
 that in the $(\min,+)$ vector convolution
 $(c_0,...,c_{10})$ of $a$ and $b$, in particular
 $c_4= \min\{a_0+b_4,a_1+b_3,a_2+b_2,a_3+b_1,a_4+b_0\}=11$ holds. Similarly as in Example 1,
 $a$ can be decomposed
    into three non-decreasing subsequences
    $a^1=(1,\ ,3,\ ,\ ,4) ,$ $a^2=(\ ,7,\ ,\ ,8,\ ),$
    $a^3=(\ ,\ ,\ ,9,\ ,\ )$. On the other hand, $b$
    can be decomposed into two non-increasing subsequences
    $b^1=(13,\ ,11,\ ,10,\ )$ and $b^2=(\ ,7 ,\ ,5,\ ,2)$.
    Their corresponding characteristic
    Boolean vectors are $char(a^1)=(1,0,1,0,0,1),$
    $char(a^2)=(0,1,0,0,1,0),$ $char(a^3)=(0,0,0,1,0,0)$,
    and $char(b^1)=(1,0,1,0,1,0),$ $char(b^2)=(0,1,0,1,0,1),$ respectively.
    The minimum witness of
    the entry $d_4$ in the Boolean vector convolution
    $(d_0,...,d_{10})$ of $char(b^i)$ and $char(b^j)$
    is $0$ for $i=1,\ j=1,$ $4$  for $i=2,\ j=1,$
    $1$ for $i=2,\ j=2,$ and $3$ for $i=3,\ j=2,$
    respectively. For the other combinations of
    $i\in [3]$ and $j\in [2]$, it is undefined.
Hence, $c_4$ is computed as
the minimum of  $1+10,8+13,7+5,9+7$  which is $11$ as required.

When the consecutive coordinates of the two input
$n$-dimensional integer vectors are simultaneously
non-decreasing or non-increasing the problem
of computing the their $(\min,+)$ convolution
appears to be as hard  as in the general case \cite{C23}.

\begin{fact}\cite{C23}\label{fact: ale}
  The problem of computing the $(\min,+)$ convolution
  of two integer vectors $a=(a_0,...,a_{n-1})$ and
  $b=(b_0,...,b_{n-1},)$, where the sequences $a_0,...,a_{n-1}$
  and $b_0,...,b_{n-1}$ are both non-decreasing
  or both non-increasing, is 
  equally hard as computing the convolution 
  for arbitrary $n$-dimensional integer vectors.
\end{fact}
\begin{proof}
  Let $M$ be the maximum absolute value of a coordinate in
  the $a,\ b$ vectors. Transform the vectors $a,\ b$ into
  vectors $a',\ b'$ by setting
  $a'_i=a_i+2iM$ and $b'_i=b_i+2iM$ for $i=0,...,n-1.$
  Observe that both sequences $a'_0,...,a'_{n-1}$
  and $b'_0,...,b'_{n-1}$ are non-decreasing.
  Consider the $(\min,+)$ convolutions
  $c=(c_0,...,c_{2n-2})$ of the vectors $a,\ b$ and
  $c'=(c'_0,...,c'_{2n-2})$ of the vectors $a',\ b'$.
  For $k=0,...,2n-2,$ we have
  $$c'_k=\min \{ (a_{\ell}+2\ell M)+(
  b_{k-\ell}+2(k-\ell)M)|\ell \in [ \max\{k-n+1,0\}, \min\{ k,n-1\}]\}$$
 $$ =c_{k}+2kM.$$
  Analogously, we can reduce the problem of computing
  the convolution of two arbitrary $n$-dimensional integer
  vectors to that where both input vectors form non-increasing
  sequences by using the transformation
  $a''_i=a_i-2iM$ and $b''_i=b_i-2iM$ for $i=0,...,n-1.$
\qed \end{proof}
\begin{figure}[bt!]
\begin{algorithmic}[1]
\REQUIRE two $n$-dimensional
vectors $a=(a_0,...,a_{n-1})$ and 
$b=(b_0,...,b_{n-1})$
with integer coordinates and the
decomposition of $b$  into $h$ 
uniform subsequences  $b^j$ and
a parameter $\ell \in [n].$
\ENSURE the $(\min,+)$ convolution $c = (c_0,....,c_{2n-2})$ of 
$a$ and $b$. 
\STATE  sort the coordinates of $a$
in non-decreasing order and divide them
into $\lceil n/\ell \rceil $ groups $g^i$
of size $\le \ell$
\FOR {each $i\in [\lceil n/\ell \rceil]$}
\STATE form a Boolean vector $char(g^i)$
with $n$ coordinates indicating with
ones the coordinates of $a$ covered by $g^i$
\ENDFOR
\FOR {each $j\in [h]$}
\STATE form a Boolean vector $char(b^j)$
with $n$ coordinates indicating with
ones the coordinates of $b$ covered by $b^j$
\ENDFOR
\STATE initialize the vector $c=(c_0,...,c_{2n-2})$ by setting
all its coordinates to $+\infty $
\FOR {each $j\in [h]$}
\FOR {each $i\in [\lceil n/\ell \rceil]$}
\STATE
compute the Boolean vector convolution $d^i_0,...,d^i_{2n-2}$
of $char(g^i)$ and $char(b^j)$
\ENDFOR
\FOR {$k = 0$ \TO $2n-2$}
\STATE find the smallest $i$ such that
$d^i_k\neq 0$
\STATE {\bf if} such an $i$ exists  {\bf then} find 
$a_q$ of the smallest value in $g^i$ such that $char(b^j)_{k-q}=1$
\STATE {\bf if} such an $a_q$ is defined {\bf then}
$c_k\gets \min \{
a_q+b_{k-q},c_k\}$
\ENDFOR
\ENDFOR
\STATE $c \gets (c_0,...,c_{2n-2})$
\RETURN $c$
\end{algorithmic}
\caption{An algorithm for computing the $(\min,+)$ convolution
$c$ of two $n$-dimensional integer vectors $a$ and $b$
  given with a decomposition of the sequence
  of consecutive coordinates of $b$ into $h$ uniform
  subsequences.}
\label{fig: algo4}
\end{figure}
When the entries of one of the input $n$-dimensional integer vectors
range over a relatively few distinct integers then following the
general idea of the proof of Lemma 2.1 in \cite{CH20}, we can proceed
as follows.  First, we can decompose the sequence of consecutive
coordinates of the aforementioned vector into a relatively few uniform
subsequences.  Then, we can sort the coordinates of the other vector
and divide the sorted sequence into interval groups of almost equal
size. Next, we can run Boolean vector convolution on pairs composed
of characteristic Boolean vectors covering with ones a group of the other
vector and a uniform subsequence of the first vector,
respectively. Further, using the results of the Boolean convolutions,
for a fixed uniform subsequence, for $k=0,...,2n-2,$ we can find the
group with the smallest index
containing an element whose mate
belongs to the uniform subsequence. By brute-force search in the
group, we can find a smallest element having a mate in the uniform
subsequence in order to update the computed $k$ coordinate of the
$(\min,+)$ convolution.  The algorithm is depicted in Fig. \ref{fig: algo4}.
Its correctness is implied by the following
observation based on the sorted order of the groups of the other
vector and the uniformity of considered subsequences of the first
vector.

\begin{remark}\label{rem: 4}
 Let $a=(a_0,...,a_{n-1})$ and $b=(b_0,...,b_{n-1})$
  be two $n$-dimensional integer vectors.
  Next, let $a_0,...,a_{n-1}$ be divided into
  subsequences $g^i$ such that no element 
  in $g^i$ is greater that any element in
  $g^{i+1}$ for $i=1,2,...$ Suppose that $b'$
  is a uniform subsequence of $b_0,...,b_{n-1}.$
  Then, for $k=0,...,2n-2,$ $\min \{ a_q+b_{k-q}| b_{k-q}\in b'\}$
  is equal to $\min \{a_q+b_{k-q}|a_q\in g^m \wedge b_{k-q}\in b' \},$
  where $m$ is the minimum $i$ such that there is $a_q\in g^i$
  for which $b_{k-q}\in b'.$
  \end{remark}

We obtain the following generalization of Lemma 2.1 in \cite{CH20}.

\begin{theorem}\label{theo: group}
  Let $a$ and $b$ be two $n$-dimensional integer vectors.
  Suppose that the entries of $a$ or the entries of $b$
  range over at most $h$ distinct integers.
The algorithm depicted in Fig. \ref{fig: algo4}
computes the $(\min,+)$ convolution of $a$ and $b$
in $\tilde{O}(hn^{1.5})$
time.
\end{theorem}
\begin{proof}
  Since the convolution of $a$ and $b$ is equal to that
  of $b$ and $a,$ we may assume w.l.o.g. that the entries
  of $b$ range over at most $h$ distinct integers.
  The vector $b$ can be easily decomposed into $h$
  uniform subsequences, some of them might be empty,
  say in $O(n\log h)$ time. Thus, we can use Algorithm 4
  depicted in Fig. \ref{fig: algo4}. It remains to show
  the correctness of the algorithm and estimate its complexity.

  By Remark \ref{rem: 4}, the $c_k$ coordinate output
  by the algorithm is never smaller than the corresponding
  coordinate of the $(\min,+)$ vector convolution
  of $a$ and $b.$ On the other hand, if $a_q+b_{k-q}$
  is the $k$ coordinate of the $(\min,+)$ convolution
  of $a$ and $b$ then there must be a group $g^i$
  to which $a_q$ belongs and a uniform subsequence $b^j$
  that $b_{k-q}$ belongs to. The algorithm sets
  $c_k$ to a value not less than
  $\min \{ a_{q'}+b_{k-q'}| a_{q'}\in g^i \wedge b_{k-q'}\in b^j\}$
  so $c_k$ is not greater than
  the corresponding coordinate of the $(\min,+)$
  convolution.

  The coordinates of $a$ can be sorted and divided
  into the consecutive groups of $\le \ell$ elements
  in $O(n\log n)$ time. All other steps outside the block
  between lines 9 and 18 can be implemented in $O(n)$ time.
  The computations of the Boolean vector convolutions in line 11
  takes $\tilde{O}(hn/\ell \times n)$ time in total by Fact \ref{fac: 1}.
  Finding the smallest $i$ such that $d^i_k\neq 0$ in line 14 takes
  $O(hn^2/\ell)$ time in total. 
  Finally, finding $a_{q}$ of smallest value in $g^i$
  satisfying the conditions in line 15 takes $O(h n \ell)$
  time in total. By setting $\ell= \sqrt {n},$ we obtain the upper bound
  $\tilde{O} (hn^{1.5})$ on the running time of the algorithm.
    \qed
\end{proof}
Because of the correction of Theorem 3.7 from \cite{LP18} and Theorem \ref{theo: group},
several entries in Table 1
in \cite{LP18} need to be updated. Table 2
presents the updated  version of the table.

\section*{Acknowledgments}
  Thanks go to Alejandro Cassis for pointing the flaw
  in the statement of Theorem 3.7 in \cite{LP18},  providing
  Fact \ref{fact: ale}, and
  suggesting Theorem \ref{theo: group}.

\vfill
\end{document}